\documentclass[10pt]{llncs}
\usepackage{enumerate,amsmath,amssymb,mathrsfs}
\usepackage{graphicx}
\usepackage{color}
\usepackage[all]{xy}

\newcommand{\restrict}{\ensuremath{\upharpoonright}}

\renewcommand{\phi}{\varphi}

\newcommand{\ignore}[1]{}

\begin{document}

\title{On the Chen Conjecture regarding the complexity of QCSPs}
\author{Barnaby Martin\inst{1}}
\institute{
  School of Science and Technology, Middlesex University, \\
  The Burroughs, Hendon, London, NW4 4BT, UK.\\
  \url{barnabymartin@gmail.com}\\
}
\maketitle

\begin{abstract}
Let $\mathbb{A}$ be an idempotent algebra on a finite domain. We combine results of Chen \cite{AU-Chen-PGP} and Zhuk \cite{ZhukGap2015} to argue that if $\mathrm{Inv}(\mathbb{A})$ satisfies the polynomially generated powers property (PGP), then QCSP$(\mathrm{Inv}(\mathbb{A}))$ is in NP. We then use the result of Zhuk to prove a converse, that if $\mathrm{Inv}(\mathbb{A})$ satisfies the exponentially generated powers property (EGP), then QCSP$(\mathrm{Inv}(\mathbb{A}))$ is co-NP-hard. Since Zhuk proved that only PGP and EGP are possible, we derive a full dichotomy for the QCSP, justifying the moral correctness of what we term the Chen Conjecture (see \cite{Meditations}).
\end{abstract}

\section{Introduction}

This note is a place-holder giving details of the resolution of a revised form of the Chen Conjecture from \cite{Meditations}. The form we speak of is one involving \emph{infinite signatures} and co-NP-hardness in place of Pspace-hardness. This form is apposite because we have already the tools to solve it. The exposition for the NP membership side of the proof (Theorem~\ref{thm:easy}) is somewhat vague, as it requires quite a bit of machinery from \cite{LICS2015}. In this note we only give details of the simple modification to the proof from that paper, in the form of an additional lemma, that is required for infinite signatures.

The following is the merger of Conjectures 6 and 7 in \cite{Meditations} which we call the \emph{Chen Conjecture}.
\begin{conjecture}[Chen Conjecture]
Let $\mathcal{B}$ be a finite relational structure expanded with all constants. If Pol$(\mathcal{B})$ has PGP, then QCSP$(\mathcal{B})$ is in NP; otherwise QCSP$(\mathcal{B})$ is Pspace-complete.
\end{conjecture}
In \cite{Meditations}, Conjecture 6 gives the NP membership and Conjecture 7 the Pspace-completeness. We now know from \cite{ZhukGap2015} and \cite{LICS2015} that the NP membership of Conjecture 6 is indeed true.  

The main result of this paper is the following.
\begin{theorem}[Revised Chen Conjecture]
\label{thm:all}
Let $\mathbb{A}$ be an idempotent algebra on a finite domain $A$. If $\mathrm{Inv}(\mathbb{A})$ satisfies PGP, then QCSP$(\mathrm{Inv}(\mathbb{A}))$ is in NP. Otherwise, QCSP$(\mathrm{Inv}(\mathbb{A}))$ is co-NP-hard.
\end{theorem}

We are also able to refute the following form.
\begin{conjecture}[Alternative Chen Conjecture]
\label{thm:alternative}
Let $\mathbb{A}$ be an idempotent algebra on a finite domain $A$. If $\mathrm{Inv}(\mathbb{A})$ satisfies PGP, then for every finite signature reduct $\mathcal{B} \subseteq \mathrm{Inv}(\mathbb{A})$, QCSP$(\mathcal{B})$ is in NP. Otherwise, there exists a finite signature reduct $\mathcal{B} \subseteq \mathrm{Inv}(\mathbb{A})$ so that QCSP$(\mathcal{B})$ is co-NP-hard.
\end{conjecture}

\section{Preliminaries}

Let $A$ be a finite domain of size $n$ whose elements are named by constants $a_1,\ldots,a_n$. When we consider structures $\mathcal{A}$ over the set $A$ with an \emph{infinite signature}, \mbox{i.e.} an infinite set of relations, for which we wish to consider the CSP or QCSP, then there comes the question as to how these relations should be encoded. One possibility is to list all the tuples of the relation but this can be extremely lengthy. A more natural possibility is as some quantifier-free formula $\phi$ in the language of equality involving the constants $a_1,\ldots,a_n$. For example, $x \neq y \wedge x = a_1$ defines an $n$-clique in which $a_1$ has a self-loop and all other elements are loopless. If we permit this $\phi$ to be in arbitrary form, even in CNF, then testing non-emptiness of the relation can already be NP-complete, whose great overhead of complexity is rather unsatisfying. Most natural is to insist on DNF where, inter alia, non-emptiness becomes tractable. This approach has been taken, e.g., in \cite{BodDalJournal}. In this paper we will always assume that the relations are encoded in DNF; however our results would not change if we allowed arbitrary formulas since we are interested in the distinction NP versus co-NP-hard and not fine-grained analysis within NP.

Let Pol$(\mathcal{B})$ be the set of polymorphisms of $\mathcal{B}$ and let Inv$(\mathbb{A})$ be the set of relations on $A$ which are invariant under (each of) the operations of $\mathbb{A}$. Pol$(\mathcal{B})$ is an object known in Universal Algebra as a \emph{clone}, which is a set of operations containing all projections and closed under composition (superposition). I will conflate sets of operations over the same domain and algebras just as I do sets of relations over the same domain and constraint languages (relational structures). Indeed, the only technical difference between such objects is the movement away from an ordered signature, which is not something we will ever need.

For a finite-domain algebra $\mathbb{A}$ we associate a function
$f_\mathbb{A}:\mathbb{N}\rightarrow\mathbb{N}$, giving the cardinality
of the minimal generating sets of the sequence $\mathbb{A},
\mathbb{A}^2, \mathbb{A}^3, \ldots$ as $f(1), f(2), f(3), \ldots$,
respectively. We may say $\mathbb{A}$ has the $g$-GP if $f(m) \leq
g(m)$ for all $m$. The question then arises as to the growth rate of
$f$, for example, regarding the behaviours constant, logarithmic,
linear, polynomial and exponential.  We say
$\mathbb{A}$  enjoys the \emph{polynomially generated powers} property
(PGP) if there exists a polynomial $p$ so that $f_{\mathbb{A}}=O(p)$
and  the \emph{exponentially generated powers} property (EGP) if there
exists a constant $b$ so that $f_{\mathbb{A}}=\Omega(g)$ where
$g(i)=b^i$.

For a finite-domain, idempotent algebra $\mathbb{A}$, \emph{$k$-collapsibility} may be seen as a special form of the PGP in which the generating set for $\mathbb{A}^m$ is constituted of all tuples $(x_1,\ldots,x_m)$ in which at least $m-k$ of these elements are equal. \emph{$k$-switchability} may be seen as another special form of the PGP in which the generating set for $\mathbb{A}^m$ is constituted of all tuples $(x_1,\ldots,x_m)$ in which there exists $a_i<\ldots<a_{k'}$, for $k'\leq k$, so that
\[ (x_1,\ldots,x_m) = (x_1,\ldots,x_{a_1},x_{a_1+1},\ldots,x_{a_2},x_{a_2+1},\ldots,\ldots,x_{a_k'},x_{a_k'+1},\ldots,x_m),\]
where $x_1=\ldots=x_{a_1-1}$, $x_{a_1}=\ldots=x_{a_2-1}$, \ldots, $x_{a_{k'}}=\ldots=x_{a_m}$. Thus, $a_1,a_2,\ldots,a_{k'}$ are the indices where the tuple switches value. Note that these are not the original definitions but they are proved equivalent to the original definitions, at least for finite signatures, in \cite{LICS2015}. We say that $\mathbb{A}$ is collapsible (switchable) if there exists $k$ such that it is $k$-collapsible ($k$-switchable). Note that Zhuk uses our definition of switchability in \cite{ZhukGap2015} and so it is the version we will use in this paper.

\section{Results}

The following result is essentially a corollary of the works of Chen and Zhuk \cite{AU-Chen-PGP,ZhukGap2015} via \cite{LICS2015}.
\begin{theorem}
\label{thm:easy}
Let $\mathbb{A}$ be an idempotent algebra on a finite domain $A$. If $\mathrm{Inv}(\mathbb{A})$ satisfies PGP, then QCSP$(\mathrm{Inv}(\mathbb{A}))$ reduces to a polynomial number of intances of CSP$(\mathrm{Inv}(\mathbb{A}))$ and is in NP.
\end{theorem}
\begin{proof}
We know from Theorem 7 in \cite{ZhukGap2015} that $\mathbb{A}$ is switchable, whereupon it follows from 
Corollary 38 in \cite{LICS2015} that QCSP$(\mathrm{Inv}(\mathbb{A}))$ reduces to a polynomial number of instances of CSP$(\mathrm{Inv}(\mathbb{A}))$ and is in NP, with the following proviso. The proof of Corollary 38 requires that the signature be finite and an additional lemma is needed to handle an infinite signature, and this is Lemma~\ref{lem:infinite-composibility} given later and discussed further in Section~\ref{sec:added}.
\end{proof}
Note that Chen's original definition of switchability, based on adversaries and reactive composability, plays a key role in the NP membership algorithm in Theorem~\ref{thm:easy}. It is the result from \cite{LICS2015} that is required to reconcile the two definitions of switchability as equivalent, and indeed Lemma~\ref{lem:infinite-composibility} is needed in this process for infinite signatures. If we were to use just our definition of switchability then it is only possible to prove, \`a la Proposition 3.3 in \cite{AU-Chen-PGP}, that the bounded alternation $\Pi_n$-QCSP$(\mathrm{Inv}(\mathbb{A}))$ is in NP.

Suppose there exists $\alpha,\beta$ strict subsets of $A$ so that $\alpha \cup \beta = A$, define the relation $\tau_k(x_1,y_1,z_1\ldots,x_k,y_k,z_k)$ defined by
\[ \tau_k(x_1,y_1,z_1\ldots,x_k,y_k,z_k):=\rho'(x_1,y_1,z_1) \vee \ldots \vee \rho'(x_k,y_k,z_k),\]
where $\rho'(x,y,z)=(\alpha \times \alpha \times \alpha) \cup (\beta \times \beta \times \beta)$.
\begin{theorem}
\label{thm:hard}
Let $\mathbb{A}$ be an idempotent algebra on a finite domain $A$. If $\mathrm{Inv}(\mathbb{A})$ satisfies EGP, then QCSP$(\mathrm{Inv}(\mathbb{A}))$ is in co-NP-hard.
\end{theorem}
\begin{proof}
We know from Lemma 11 in \cite{ZhukGap2015} that there exist $\alpha,\beta$ strict subsets of $A$ so that $\alpha \cup \beta = A$ and the relation $\sigma_k(x_1,y_1,\ldots,x_k,y_k)$ defined by
\[ \sigma_k(x_1,y_1,\ldots,x_k,y_k):=\rho(x_1,y_1) \vee \ldots \vee \rho(x_k,y_k),\]
where $\rho(x,y)=(\alpha \times \alpha) \cup (\beta \times \beta)$, is in $\mathrm{Inv}(\mathbb{A})$, for each $n \in \mathbb{N}$.

We will first argue that the relation $\tau_k(x_1,y_1,z_1\ldots,x_k,y_k,z_k)$ is in $\mathrm{Inv}(\mathbb{A})$, for each $k \in \mathbb{N}$. For this it is enough to see that $\tau_k$ is definable by the conjunction $\Phi$ of $3^n$ instances of $\sigma_k$ that each consider the ways in which two variables may be chosen from each of the $(x_i,y_i,z_i)$, i.e. $x_i=y_i$ or $y_i=z_i$ or $x_i=z_i$. We need to show that this conjunction $\Phi$ entails $\tau_n$ (the converse is trivial). We will assume for contradiction that $\Phi$ is satisfiable but $\tau_n$ not. In the first instance of $\sigma_n$ of $\Phi$ some atom must be true, and it will be of the form $x_i=y_i$ or $y_i=z_i$ or $x_i=z_i$. Once we have settled on one of these three, $p_i=q_i$, then we immediately satisfy $3^{k-1}$ of the conjunctions of $\Phi$, leaving $2\cdot 3^{k-1}$ unsatisfied. Now, we can not evaluate true any of the others among $\{x_i=y_i, y_i=z_i, x_i=z_i\} \setminus \{p_i=q_i\}$ without contradicting our assumption. Thus we are now down to looking at variables with subscript other than $i$ and in this fashion we have made the space one smaller, in total $k-1$. Now, we will need to evaluate in $\Phi$ some other atom of the form  $x_j=y_j$ or $y_j=z_j$ or $x_j=z_j$, for $j\neq i$. Once we have settled on one of these three then we immediately satisfy $2 \cdot 3^{k-2}$ of the conjunctions remaining of $\Phi$, leaving $2^2 \cdot 3^{k-2}$ still unsatisfied. Iterating this thinking, we arrive at a situation in which $2^k$ clauses are unsatisfied after we have gone through all $k$ subscripts, which is a contradiction. 

We will next argue that $\tau_k$ enjoys a relatively small specification in DNF (at least, polynomial in $k$). We first give such a specification for $\rho'(x,y,z)$.
\[ \rho'(x,y,z):= \bigvee_{a,a',a'' \in \alpha} x=a \wedge y=a' \wedge z=a'' \vee \bigvee_{b,b',b'' \in \beta} x=b \wedge y=b' \wedge z=b''\]
which is constant in size when $A$ is fixed. Now it is clear from the definition that the size of $\tau_n$ is polynomial in $n$.

We will now give a very simple reduction from the complement of (monotone) \emph{$3$-not-all-equal-satisfiability} (3NAESAT) to QCSP$(\mathrm{Inv}(\mathbb{A}))$. 3NAESAT is well-known to be NP-complete \cite{Papa} and our result will follow.

Take an instance $\phi$ of 3NAESAT which is the existential quantification of a conjunction of $k$ atoms $\mathrm{NAE}(x,y,z)$. Thus $\neg \phi$ is the universal quantification of a disjunction of $k$ atoms $x=y=z$. We build our instance $\psi$ of QCSP$(\mathrm{Inv}(\mathbb{A}))$ from $\neg \phi$ by transforming the quantifier-free part $x_1=y_1=z_1 \vee \ldots \vee x_k=y_k=z_k$ to $\tau_k=\rho'(x_1,y_1,z_1) \vee \ldots \vee \rho'(x_k,y_k,z_k)$.

($\neg \phi \in \mathrm{co\mbox{-}3NAESAT}$ implies $\psi \in \mathrm{QCSP}(\mathrm{Inv}(\mathbb{A}))$.) From an assignment to the universal variables $v_1,\ldots,v_m$ of $\psi$ to elements $x_1,\ldots,x_m$ of $A$, consider elements $x'_1,\ldots,x'_m \in \{0,1\}$ according to 
\begin{itemize}
\item $x_i \in \alpha \setminus \beta$ implies $x'_i=0$, 
\item $x_i \in \beta \setminus \alpha$ implies $x'_i=1$, and
\item $x_i \in \alpha \cap \beta$ implies we don't care, so w.l.o.g. say $x'_i=0$.
\end{itemize}
The disjunct that is satisfied in the quantifier-free part of $\neg \phi$ now gives the corresponding disjunct that will be satisfied in $\tau_k$.

($\psi \in \mathrm{QCSP}(\mathrm{Inv}(\mathbb{A}))$ implies $\neg \phi \in \mathrm{co\mbox{-}3NAESAT}$.) From an assignment to the universal variables $v_1,\ldots,v_m$ of $\phi$ to elements $x_1,\ldots,x_m$ of $\{0,1\}$, consider elements $x'_1,\ldots,x'_m \in A$ according to 
\begin{itemize}
\item $x_i=0$ implies $x'_i$ is some arbitrarily chosen element in $\alpha \setminus \beta$, and
\item $x_i=1$ implies $x'_i$ is some arbitrarily chosen element in $\beta \setminus \alpha$.
\end{itemize}
The disjunct that is satisfied in $\tau_k$ now gives the corresponding disjunct that will be satisfied in the quantifier-free part of $\neg \phi$.
\end{proof}
\noindent The demonstration of co-NP-hardness in the previous theorem was inspired by a similar proof in \cite{BodirskyChenSICOMP}.\footnote{Indeed one may say the only new insight for our main result of the Revised Chen Conjecture is the move to an infinite signature together with the acceptance of co-NP-hardness over Pspace-completeness.} 

We note surprisingly that co-NP-hardness in Theorem~\ref{thm:hard} is optimal.
\begin{proposition}
Let $\alpha,\beta$ strict subsets of $A:=\{a_1,\ldots,a_n\}$ so that $\alpha \cup \beta = A$ and $\alpha \cap \beta \neq \emptyset$. Then QCSP$(A;\{\tau_k:k \in \mathbb{N}\},a_,\ldots,a_n)$ is in co-NP.
\end{proposition}
\begin{proof}
Let $\phi$ be an input to QCSP$(A;\{\tau_k:k \in \mathbb{N}\},a_1,\ldots,a_n)$. We will now seek to eliminate atoms $v=a$ ($a \in \{a_1,\ldots,a_n\}$) from $\phi$. Suppose $\phi$ has an atom $v=a$. If $v$ is universally quantified, then $\phi$ is false. Otherwise, either the atom $v=a$ may be eliminated with the variable $v$ since $v$ does not appear in a non-equality relation; or $\phi$ is false because there is another atom $v=a'$ for $a\neq a'$; or $v=a$ may be removed by substitution of $a$ into all non-equality instances of relations involving $v$. This preprocessing procedure is polynomial and we will assume \mbox{w.l.o.g.} that $\phi$ contains no atoms $v=a$. We now argue that $\phi$ is a yes-instance iff $\phi'$ is a yes-instance, where $\phi'$ is built from $\phi$ by instantiating all existentially quantified variables as any $a \in \alpha \cap \beta$. The universal $\phi'$ can be evaluated in co-NP (one may prefer to imagine the complement as an existential $\neg \phi' $ to be evaluated in NP) and the result follows.
\end{proof}
The following, together with our previous results, gives the refutation of the Alternative Chen Conjecture.
\begin{proposition}
Let $\alpha,\beta$ strict subsets of $A:=\{a_1,\ldots,a_n\}$ so that $\alpha \cup \beta = A$ and $\alpha \cap \beta \neq \emptyset$. Then, for each finite signature reduct $\mathcal{B}$ of $(A;\{\tau_k:k \in \mathbb{N}\},a_,\ldots,a_n)$, QCSP$(\mathcal{B})$ is in NL.
\end{proposition}
\begin{proof}
We will assume $\mathcal{B}$ contains all constants (since we prove this case gives a QCSP in NL, it naturally follows that the same holds without constants). Take $n$ so that, for each $\tau_i \in \mathcal{B}$, $i\leq n$. We claim $\mathcal{B}$ admits a $(3n+1)$-ary near-unanimity operation $f$ as a polymorphism, where all tuples whose map is not automatically defined by the near-unanimity criterion map to some arbitrary $a \in \alpha \cap \beta$. To see this, imagine that this $f$ were not a polymorphism. Then some $(3n+1)$ tuples in $\tau_n$ would be mapped to some tuple not in $\tau_n$ which must be a tuple $\overline{t}$ of elements from $\alpha \setminus \beta \cup \beta \setminus \alpha$. Note that column-wise this map may only come from $(3n+1)$-tuples that have $3n$ instances of the same element. By the pigeonhole principle, the tuple $\overline{t}$ must appear as one of the $(3n+1)$ tuples in $\tau_n$ and this is clearly a contradiction.

It follows from \cite{hubie-sicomp} that QCSP$(\mathcal{B})$ reduces to a polynomially bounded ensemble of instances CSP$(\mathcal{B})$, and the result follows.
\end{proof}

\section{Canonical sentences on infinite signatures}
\label{sec:added}

The reader will need to refresh their knowledge of the terminology from \cite{LICS2015} or better the longer arxiv version \cite{LICS2015-arxiv} to which we will refer in this section. If $\mathcal{A}$ is a $\sigma$-structure and $\sigma' \subseteq \sigma$, then let $\mathcal{A}^{\sigma'}$ be the $\sigma'$ reduct of $\mathcal{A}$.

A \emph{canonical sentence for composability for arbitrary pH-sentences} with $m$ universal variables
may be constructed similarly to the canonical sentence for the $\Pi_2$ case, except that it will have $m.n$ universal variables, which we view as $m$ blocks of $n$ variables, where $n$ is the number of elements of the structure $\mathcal{A}$.
Let $\mathscr{O}$ be some adversary of length $m$.
Let $\sigma^{(n\cdot m)}$ be the signature $\sigma$ expanded with a sequence of $n.m$ constants $c_{1,1},\ldots, c_{n,1}, c_{1,2} \ldots, c_{n,2}, \ldots c_{1,m} \ldots, c_{n,m}$. We say that a map $\mu$ from $[n]\times [m]$ to $A$ is \emph{consistent} with $\mathscr{O}$ iff for every $(i_1,i_2,\ldots,i_m)$ in $[n]^m$, the tuple $(\mu(i_1,1),\mu(i_2,2),\ldots,\mu(i_m,m))$ belongs to the adversary $\mathscr{O}$. We write $A^{[n.m]}_{\restrict \mathscr{O}}$ for the set of such consistent maps.
For some set $\Omega_m$ of adversaries of length $m$, we consider the following $\sigma^{(n.m)}$-structure: 
$$\bigotimes_{\mathscr{O}\in \Omega_m} \bigotimes_{\mu \in A^{[n.m]}_{\restrict \mathscr{O}}} \mathfrak{A}_{\mathscr{O},\mu}$$
where the $\sigma^{(n\cdot m)}$-structure $\mathfrak{A}_{\mathscr{O},\mu}$ denotes the expansion of $\mathcal{A}$ by $n.m$ constants as given by the map $\mu$.
Let $\phi_{n,\Omega_m,\mathcal{A}}$ be the \textbf{infinite} $\Pi_2$-pH-sentence created from the canonical query of the $\sigma$-reduct of this $\sigma^{(n.m)}$ product structure with the $n.m$ constants $c_{ij}$ becoming variables $w_{ij}$, universally quantified outermost.
This sentence is
not well defined if constants are not pairwise distinct, which occurs precisely
for degenerate adversaries. 

Now, for each finite subset $\sigma' \subset \sigma$ we are also interested in the similar  $\sigma'^{(n.m)}$-structure: 
$$\bigotimes_{\mathscr{O}\in \Omega_m} \bigotimes_{\mu \in A^{[n.m]}_{\restrict \mathscr{O}}} \mathfrak{A}^{\sigma'}_{\mathscr{O},\mu}$$
where the $\sigma'^{(n\cdot m)}$-structure $\mathfrak{A}^{\sigma'}_{\mathscr{O},\mu}$ denotes the expansion of $\mathcal{A}$, restricted to the signature $\sigma'$, by $n.m$ constants as given by the map $\mu$.
Let $\phi^{\sigma'}_{n,\Omega_m,\mathcal{A}}$ be the \textbf{finite} $\Pi_2$-pH-sentence created from the canonical query of the $\sigma'$-reduct of this $\sigma'^{(n.m)}$ product structure with the $n.m$ constants $c_{ij}$ becoming variables $w_{ij}$, universally quantified outermost.

The following lemma gives a form of compactness.
\begin{lemma}
\label{lem:infinite-composibility}
Let $\mathcal{A}$ be a finite-domain structure with domain size $n$ on an infinite signature $\sigma$. For each $m$, $\mathcal{A} \models \phi_{n,\Omega_m,\mathcal{A}}$ iff for every finite subset $\sigma' \subset \sigma$, $\mathcal{A}^{\sigma'} \models \phi^{\sigma'}_{n,\Omega_m,\mathcal{A}}$.
\end{lemma}
\begin{proof}
The forward direction is trivial and we will prove the backwards direction by contradiction. Suppose there is an assignment of $a_{ij}$ to the universal variables $w_{ij}$ that witnesses the falsehood of $\phi_{n,\Omega_m,\mathcal{A}}$ on $\mathcal{A}$. Suppose this assignment can be extended for truth however to $a^{\sigma'}_1,\ldots,a^{\sigma'}_r$ on the existential variables of $\phi^{\sigma'}_{n,\Omega_m,\mathcal{A}}$ for each $\mathcal{A}^{\sigma'}$. Now, some tuple $(a_1,\ldots,a_r)$ must appear infinitely often among the $a^{\sigma'}_1,\ldots,a^{\sigma'}_r$ but it would necessarily witness truth for $\phi_{n,\Omega_m,\mathcal{A}}$ on $\mathcal{A}$, and this is a contradiction.
\end{proof}
We need Lemma~\ref{lem:infinite-composibility} in order to derive the version of Lemma 34 in \cite{LICS2015-arxiv} for infinite signatures. Note that Corollary 38 in \cite{LICS2015,LICS2015-arxiv} does not mention switchability by name, but switchability is an example of an effective, projective and polynomially-bounded adversary.

\section{Discussion}

Note that the aggregation of Theorems~\ref{thm:easy} and \ref{thm:hard}, from which Theorem~\ref{thm:all} follows, actually says something stronger. It shows that the finite QCSP classification, at least in the idempotent cases (when all constants are present), now reduces to the hitherto unknown CSP classification. That is, either we know  $\mathrm{QCSP}(\mathrm{Inv}(\mathbb{A}))$ is co-NP-hard, or we need to read its complexity from that of $\mathrm{CSP}(\mathrm{Inv}(\mathbb{A}))$. Conversely, it is well-known that the CSP classification embeds in the QCSP classification. Thus, the remaining work for the QCSP classification is precisely that for the CSP classification. The only proviso here is that we are dealing with infinite signatures, something unusual in the literature.

A similar reduction from the Valued CSP classification to the CSP classification was recently achieved, in a series of papers culminating in \cite{FOCS2015}. 

The Chen Conjecture in its original form remains open.

\bibliographystyle{acm}

\begin{thebibliography}{10}

\bibitem{BodirskyChenSICOMP}
{\sc Bodirsky, M., and Chen, H.}
\newblock Quantified equality constraints.
\newblock {\em SIAM J. Comput. 39}, 8 (2010), 3682--3699.

\bibitem{BodDalJournal}
{\sc Bodirsky, M., and Dalmau, V.}
\newblock Datalog and constraint satisfaction with infinite templates.
\newblock {\em Journal on Computer and System Sciences 79\/} (2013), 79--100.
\newblock A preliminary version appeared in the proceedings of the Symposium on
  Theoretical Aspects of Computer Science (STACS'05).

\bibitem{LICS2015}
{\sc Carvalho, C., Madelaine, F.~R., and Martin, B.}
\newblock From complexity to algebra and back: digraph classes, collapsibility
  and the {PGP}.
\newblock In {\em 30th Annual IEEE Symposium on Logic in Computer Science
  (LICS)\/} (2015).

\bibitem{LICS2015-arxiv}
{\sc Carvalho, C., Madelaine, F.~R., and Martin, B.}
\newblock From complexity to algebra and back: digraph classes, collapsibility
  and the {PGP}, 2015.
\newblock http://arxiv.org/abs/1501.04558.

\bibitem{hubie-sicomp}
{\sc Chen, H.}
\newblock The complexity of quantified constraint satisfaction: Collapsibility,
  sink algebras, and the three-element case.
\newblock {\em SIAM J. Comput. 37}, 5 (2008), 1674--1701.

\bibitem{AU-Chen-PGP}
{\sc Chen, H.}
\newblock Quantified constraint satisfaction and the polynomially generated
  powers property.
\newblock {\em Algebra universalis 65}, 3 (2011), 213--241.
\newblock An extended abstract appeared in ICALP B 2008.

\bibitem{Meditations}
{\sc Chen, H.}
\newblock Meditations on quantified constraint satisfaction.
\newblock In {\em Logic and Program Semantics - Essays Dedicated to Dexter
  Kozen on the Occasion of His 60th Birthday\/} (2012), pp.~35--49.

\bibitem{FOCS2015}
{\sc Kolmogorov, V., Krokhin, A.~A., and Rolinek, M.}
\newblock The complexity of general-valued csps.
\newblock In {\em {IEEE} 56th Annual Symposium on Foundations of Computer
  Science, {FOCS} 2015, Berkeley, CA, USA, 17-20 October, 2015\/} (2015),
  pp.~1246--1258.

\bibitem{Papa}
{\sc Papadimitriou, C.~H.}
\newblock {\em Computational Complexity}.
\newblock Addison-Wesley, 1994.

\bibitem{ZhukGap2015}
{\sc {Zhuk}, D.}
\newblock {The Size of Generating Sets of Powers}.
\newblock {\em ArXiv e-prints\/} (Apr. 2015).

\end{thebibliography}

\end{document}